\documentclass[12pt, final]{l4dc2022}

\usepackage{bm}

\DeclareMathAlphabet\mathbfcal{OMS}{cmsy}{b}{n}

\DeclareMathOperator*{\st}{subject~to}

\DeclareMathOperator*{\argmin}{arg\,min}

\definecolor{darkgreen}{rgb}{0.0, 0.4, 0.0}
\definecolor{lightgray}{gray}{0.9}

\newcommand{\norm}[1]{\left\lVert#1\right\rVert}

\title[Safe Control with Minimal Regret]{Safe Control with Minimal Regret}
\usepackage{times}

\author{
 \Name{Andrea Martin} \Email{andrea.martin@epfl.ch}\\
 \addr Institute of Mechanical Engineering, \'Ecole Polytechnique F\'ed\'erale de Lausanne, Switzerland
 \AND
 \Name{Luca Furieri} \Email{luca.furieri@epfl.ch}\\
 \addr Institute of Mechanical Engineering, \'Ecole Polytechnique F\'ed\'erale de Lausanne, Switzerland
 \AND
 \Name{Florian D\"orfler} \Email{dorfler@ethz.ch}\\
 \addr Department of Information Technology and Electrical Engineering, ETH Z\"urich, Switzerland
 \AND
 \Name{John Lygeros} \Email{jlygeros@ethz.ch}\\
 \addr Department of Information Technology and Electrical Engineering, ETH Z\"urich, Switzerland
 \AND
 \Name{Giancarlo Ferrari-Trecate} \Email{giancarlo.ferraritrecate@epfl.ch}\\
 \addr Institute of Mechanical Engineering, \'Ecole Polytechnique F\'ed\'erale de Lausanne, Switzerland
}

\begin{document}

\maketitle

\begin{abstract}
    As we move towards safety-critical cyber-physical systems that operate in non-stationary and uncertain environments, it becomes crucial to close the gap between classical optimal control algorithms and adaptive learning-based methods. In this paper, we present an efficient optimization-based approach for computing a finite-horizon robustly safe control policy that minimizes dynamic regret, in the sense of the loss relative to the optimal sequence of control actions selected in hindsight by a clairvoyant controller. By leveraging the system level synthesis framework (SLS), our method extends recent results on regret minimization for the linear quadratic regulator to optimal control subject to hard safety constraints, and allows competing against a safety-aware clairvoyant policy with minor modifications. Numerical experiments confirm superior performance with respect to finite-horizon constrained $\mathcal{H}_2$ and $\mathcal{H}_\infty$ control laws when the disturbance realizations poorly fit classical assumptions.
\end{abstract}


\section{Introduction}
\label{sec:introduction}
Classical $\mathcal{H}_2$ and $\mathcal{H}_\infty$ control theories have studied how to optimally regulate the behavior of a linear dynamical system driven by a stochastic or worst-case disturbance process, respectively. In the $\mathcal{H}_2$ setting, control actions are chosen so to minimize the expected control cost incurred by the learner. Conversely, in the $\mathcal{H}_\infty$ setting, the agent minimizes the worst-case control cost across all disturbance realizations \citep{zhou1998essentials}. Both paradigms commit to a control strategy that is tailored to the presumed nature of the disturbance process, and that is blindly pursued regardless of the observed perturbation sequence. Hence, these control policies may suffer high cumulative costs if the disturbance realizations do not match the class of perturbations hypothesized a priori.

Classical literature on online learning has studied sequential decision-making algorithms that learn from experience, as repeated interactions with a \textit{memoryless} environment reveal more aspects of the problem at hand. In this setting, agents make no assumptions about the benign or adversarial nature of the environment, and they dynamically adjust their strategy based on information deduced from previous rounds to minimize regret \citep{shalev2011online, hazan2016introduction}. Informally, one can think of regret as measuring the loss suffered by a learner relative to the optimal policy in hindsight. 

When the environment evolves dynamically, online learning comes with significant new challenges, since the losses agents observe depend not only on their current action, but also on their past decisions. Recent years have witnessed an increasing interest in applying modern statistical and algorithmic techniques to classical control problems. Initiated by \cite{abbasi2011regret}, several works have approached the problem of adaptively controlling a linear dynamical system perturbed by a stochastic \citep{dean2018regret, cohen2019learning, lale2020logarithmic} or adversarial \citep{agarwal2019online, foster2020logarithmic, hazan2020nonstochastic, simchowitz2020improper} disturbance process from the perspective of \textit{policy regret} minimization. In this framework, control algorithms are designed to compete with the best \textit{static} policy selected in hindsight from a parametric class, and performance guarantees are expressed in terms of sublinear regret bounds against this idealized benchmark. Intuitively, attaining sublinear regret is desirable, as it implies that the average difference between the cost suffered by the learning algorithm and that of the best fixed a posteriori strategy converges to zero over time.

When the online learner interacts with a dynamic environment, policy regret minimization algorithms may yield loose performance certificates \citep{hazan2009efficient}. Indeed, while these methods effectively approach the best fixed strategy in the chosen comparator class, this static benchmark could also incur a high control cost. For instance, no single state-feedback controller can perform well in a scenario where disturbances alternate between being drawn according to a well-behaved stochastic process and being chosen adversarially \citep{goel2020regret}. 

This consideration motivates the design of algorithms that compete against the \textit{clairvoyant} control law that selects the globally optimal dynamic sequence of control actions in hindsight \citep{goel2020power}. This concept of \textit{dynamic regret}, which was first introduced in \cite{zinkevich2003online}, has recently been studied in the full-information setting in \cite{goel2020regret, sabag2021regret}, and in the measurement-feedback setting in \cite{goel2021regret}. Assuming knowledge of the underlying dynamics and control costs, these works have not only provided tight dynamic regret bounds in terms of the energy of the disturbance sequence,\footnote{Dynamic regret bounds are often expressed in terms of some ``regularity'' of the perturbation sequence since dynamic regret scales linearly with time in the worst-case \citep{jadbabaie2015online, goel2020power, zhao2020dynamic}.} but have also computed a control law that exactly minimizes the worst-case loss relative to the clairvoyant optimal policy; specifically, regret-optimal controllers have been explicitly derived, both in state-space and input-output form, via elegant reductions to classical $\mathcal{H}_\infty$ synthesis and Nehari noncausal approximation problems \citep{nehari1957bounded}.

Dynamic regret minimization algorithms, which combine design criteria from online learning theory with continuous action spaces typical of control, have further been shown to nicely interpolate between the performance of classical $\mathcal{H}_2$ and $\mathcal{H}_\infty$ controllers across both stochastic and adversarial environments \citep{goel2020regret, sabag2021regret}. However, due to the lack of provable robustness guarantees, these control policies do not lend themselves to real-time applications that impose hard safety constraints on the physical variables of the system. Despite being critical to reliably deploy learning-based methods in modern engineering systems, the problem of designing safety-aware algorithms with non-asymptotic performance guarantees has been only recently approached in \cite{nonhoff2021online}, which assumes that the underlying linear dynamics are not affected by disturbances, and in \cite{li2020online, li2021safe}, albeit from the standpoint of policy regret minimization. To the best of our knowledge, the literature offers no characterization of a control policy that simultaneously attains minimum dynamic regret and guarantees compliance with hard safety constraints in face of the uncertain disturbance realizations.

In this paper, we present an efficient optimization-based approach for computing a finite-horizon safe regret-optimal control policy that exactly minimizes the loss relative to the clairvoyant optimal controller, while satisfying safety constraints defined over physical variables of the system. To do so, we leverage the recent system level parametrization (SLP) of linear dynamic controllers \citep{wang2019system} as well as classical duality results from robust optimization. Compared to previous works \citep{goel2020regret, sabag2021regret}, we directly treat the system \textit{closed-loop responses} as design variables, and we reduce the safe controller synthesis task to solving a semidefinite program (SDP). The proposed method circumvents the need to repeatedly appeal to the whitening property of Kalman filters for calculating explicit causal matrix factorizations \citep{kailath2000linear}, and it allows a one-shot computation of a safe regret-optimal controller, without requiring to iteratively solve suboptimal instances of the problem. As such, we believe the flexibility provided by this optimization perspective can pave the way towards analyzing more complex control problems through the lens of regret minimization. Our approach also permits to naturally capture the time-varying nature of the system dynamics, the control costs, and the safety constraints. We present numerical experiments to support using dynamic regret as control design criterion.

\section{Problem Statement and Preliminaries}
\label{sec:problem_statement_and_preliminaries}
We consider known discrete-time linear time-varying dynamical systems with state-space equations
\begin{equation}
    \label{eq:state_dynamics}
    x_{t+1} = A_t x_t + B_t u_t + w_t\,,
\end{equation}
where $u_t \in \mathbb{R}^m$, $x_t \in \mathbb{R}^n$ and $w_t \in \mathbb{R}^n$ are the control input, the system state, and an exogenous disturbance, respectively. We do not make any assumptions about the statistical distribution of the disturbance process, which we allow to be of adversarial nature. 

We study the evolution over a horizon $T \in \mathbb{N}$ of system \eqref{eq:state_dynamics}, initialized at $x_0 \in \mathbb{R}^n$, when the control actions $u_t$ are computed according to a time-varying linear feedback control policy
\begin{equation}
    \label{eq:control_policy}
    u_t = \sum_{k = 0}^t K_{t,k} x_k\,, ~ \forall t \in \{0 \dots T-1\}\,.
\end{equation}
For convenience, we compactly write signals and causal operators over $T$ as
\begin{equation*}
    \mathbf{x} = \begin{bmatrix}x_0 \\ x_1 \\ \vdots \\ x_{T-1}\end{bmatrix}\,, ~
    \mathbf{u} = \begin{bmatrix}u_0 \\ u_1 \\ \vdots \\ u_{T-1}\end{bmatrix}\,, ~
    \mathbf{w} = \begin{bmatrix}x_0 \\ w_0 \\ \vdots \\ w_{T-2}\end{bmatrix}\,, ~
    \mathbf{K} =
    \begin{bmatrix}
    	K_{0,0} & 0_{m \times n} & \dots & 0_{m \times n}\\
    	K_{1,0} & K_{1,1} & \ddots & \vdots\\
    	\vdots & \vdots & \ddots & 0_{m \times n}\\
    	K_{T-1, 0} & K_{T-1, 1}& \cdots & K_{T-1, T-1}
    \end{bmatrix}\,,
\end{equation*}
and we denote the online cost incurred applying the control sequence $\mathbf{u}$ in response to the disturbance realization $\mathbf{w}$ by 
\begin{equation}
    \label{eq:lqr_cost}
    \operatorname{cost}(\mathbf{w}, \mathbf{u}) = 
    \begin{bmatrix}
        \mathbf{x}\\
        \mathbf{u}
    \end{bmatrix}^\top
    \begin{bmatrix}
        \mathcal{Q} & 0\\
        0 & \mathcal{R}\\
    \end{bmatrix}
    \begin{bmatrix}
        \mathbf{x}\\
        \mathbf{u}
    \end{bmatrix}
    \,,
\end{equation}
where $\mathcal{Q} \succeq 0$ and $\mathcal{R} \succ 0$ represent the cost matrices coupling state and input signals at possibly different time instants, respectively.
 
Our goal is to synthesize a safe control policy in the form of \eqref{eq:control_policy} that exactly minimizes the loss relative to the clairvoyant policy, i.e., the offline control law that selects the globally optimal sequence of control actions with complete foreknowledge of the disturbance realizations. In other words, we wish to minimize the worst-case dynamic regret
\begin{equation}
    \label{eq:dynamic_regret_definition}
    \operatorname{regret}(T, \mathbf{u}) = \max_{\|\mathbf{w}\|_{2} \leq 1} ~
    \left(\operatorname{cost}(\mathbf{w}, \mathbf{u}) - \min_{\mathbf{u}^\prime \in \mathbb{R}^{mT}} ~ \operatorname{cost}\left(\mathbf{w}, \mathbf{u}^{\prime}\right)\right)\,,
\end{equation}
while complying with polytopic safety constraints defined over the system states and inputs as per
\begin{equation}
    \label{eq:safety_constraints}
    \mathcal{H} 
    \begin{bmatrix}
        \mathbf{x}^\top & \mathbf{u}^\top
    \end{bmatrix}^\top
    \leq \mathbf{h}\,, ~
    \forall \mathbf{w} \in \mathbb{W} = \left\{\mathbf{w}^\prime \in \mathbb{R}^{nT} : \mathcal{H}_w \mathbf{w}^\prime \leq \mathbf{h}_w\right\}\,,
\end{equation}
where $\mathbb{W}$ is a compact polytope that contains an open neighborhood of the origin. For the rest of the paper, inequalities involving vectors are to be intended element-wise. As it is standard in the robust control literature \citep{rawlings2017model}, we formulate the following feasibility assumption:
\paragraph{Assumption~1} There exists a causal controller $\mathbf{K}$ satisfying \eqref{eq:safety_constraints} for all perturbations $\mathbf{w} \in \mathbb{W}$.

By competing against the clairvoyant policy in \eqref{eq:dynamic_regret_definition}, we shape the worst-case gain from the disturbance energy to the regret through system-dependent optimal performance weights. Ultimately, this leads to synthesizing control laws that are adaptive, in the sense that they strive to minimize the incurred control cost independently of how the disturbances are generated. 

\begin{remark}
    We rely on closed-loop control policies of the form \eqref{eq:control_policy} to dynamically adjust to the per-instance disturbance realizations. This is beneficial not only to track the performance of the clairvoyant controller, but also to fulfill the safety constraints \eqref{eq:safety_constraints}. Selecting a priori the control actions to be applied, based on the open-loop prediction of the system evolution only, may lead to excessive conservatism, or even infeasibility \citep{bemporad1998reducing}.
\end{remark}
\subsection{System Level Synthesis}
\label{subsec:system_level_synthesis}
We briefly outline the necessary background on the SLS approach to optimal controller synthesis, and we refer to \cite{wang2019system} and \cite{anderson2019system} for a complete discussion. Akin to Youla-based and disturbance-feedback controllers \citep{youla1976modern, goulart2007affine}, the SLS approach shifts the synthesis problem from directly designing the controller to shaping the closed-loop maps from the exogenous disturbance to the state and input signals \citep{furieri2019input, zheng2020equivalence}.

Let $Z$ be the block-downshift operator, namely a matrix with identity matrices along its first block sub-diagonal and zeros elsewhere, and define $\mathcal{A} = \operatorname{blkdiag}(A_0, A_1, \dots, A_{T-2}, 0_{n \times n})$ and $\mathcal{B} = \operatorname{blkdiag}(B_0, B_1, \dots, B_{T-2}, 0_{n \times m})$. Then, the evolution of the input and state trajectories of system \eqref{eq:state_dynamics} over the control horizon can be compactly expressed as
\begin{equation*}
    \mathbf{x} = Z\mathcal{A}\mathbf{x} + Z\mathcal{B}\mathbf{u} + \mathbf{w}\,.
\end{equation*}
Similarly, the closed-loop system behavior under the feedback law $\mathbf{u} = \mathbf{K}\mathbf{x}$ can be described, as a function of the perturbation vector $\mathbf{w}$, through the non-convex relations
\begin{equation}
    \label{eq:closed_loop_behavior_K}
    \begin{split}
        \mathbf{x} &= \left(I - Z\left(\mathcal{A} +  \mathcal{B}\mathbf{K}\right)\right)^{-1} \mathbf{w} = \bm{\Phi}_x \mathbf{w}\,,\\
        \mathbf{u} &= \mathbf{K}\left(I - Z\left(\mathcal{A} +  \mathcal{B}\mathbf{K}\right)\right)^{-1} \mathbf{w} = \bm{\Phi}_u \mathbf{w}\,.
    \end{split}
\end{equation}
Let $\bm{\Phi}_x = \left(I - Z\left(\mathcal{A} +  \mathcal{B}\mathbf{K}\right)\right)^{-1}$ and $\bm{\Phi}_u = \mathbf{K}\left(I - Z\left(\mathcal{A} +  \mathcal{B}\mathbf{K}\right)\right)^{-1}$ denote the system closed-loop responses induced by the controller $\mathbf{K}$ in \eqref{eq:closed_loop_behavior_K}, and observe that these operators inherit a lower block-diagonal causal structure. One can show that there exists a controller $\mathbf{K}$ such that $\mathbf{x} = \bm{\Phi}_x\mathbf{w}$ and $\mathbf{u} = \bm{\Phi}_u\mathbf{w}$ if and only if
\begin{equation}
    \label{eq:sls_affine_subspace_constraints}
    (I-Z\mathcal{A})\bm{\Phi}_x - Z\mathcal{B}\bm{\Phi}_u = I\,;
\end{equation}
we call pairs $\{\bm{\Phi}_x, \bm{\Phi}_u\}$ that satisfy \eqref{eq:sls_affine_subspace_constraints} \textit{achievable}.

Based on \eqref{eq:sls_affine_subspace_constraints}, many optimal control problems of practical interest can be equivalently posed as an optimization over the convex set of system responses $\{\bm{\Phi}_x, \bm{\Phi}_u\}$. For instance, one can express classical unconstrained $\mathcal{H}_2$ and $\mathcal{H}_\infty$ control problems in terms of the achievable closed-loop responses by exploiting a priori assumptions on the disturbance process, i.e., $\mathbf{w} \overset{\text{iid}}{\sim} \mathcal{N}(0, \bm{\Sigma}_w)$, where $\bm{\Sigma}_w \succ 0$ denotes the covariance matrix of $\mathbf{w}$, in the $\mathcal{H}_2$ setting, or $\mathbf{w}$ adversarially chosen in the $\mathcal{H}_\infty$ setting. For the sake of comparison, we report these reformulations \citep{anderson2019system} below on the left and right sides of the page, respectively:
\begin{alignat}{3}
    &~\min_{\bm{\Phi}_x,\bm{\Phi}_u} ~ \norm{
    \begin{bmatrix}
        \mathcal{Q}^{\frac{1}{2}} & 0\\
        0 & \mathcal{R}^{\frac{1}{2}}
    \end{bmatrix}
    \begin{bmatrix}
        \bm{\Phi}_x\\\bm{\Phi}_u
    \end{bmatrix}
    \bm{\Sigma}_w^{\frac{1}{2}}
    }_{F}^2
    \hspace{2cm}
    &&~\min_{\bm{\Phi}_x,\bm{\Phi}_u} ~ \norm{
    \begin{bmatrix}
        \mathcal{Q}^{\frac{1}{2}} & 0\\
        0 & \mathcal{R}^{\frac{1}{2}}
    \end{bmatrix}
    \begin{bmatrix}
        \bm{\Phi}_x\\\bm{\Phi}_u
    \end{bmatrix}
    }_{2 \rightarrow 2}^2
    \nonumber\\
    &\st~ \begin{bmatrix}
        I-Z\mathcal{A} & -Z\mathcal{B}
    \end{bmatrix}
    \begin{bmatrix}
    \bm{\Phi}_x\\\bm{\Phi}_u
    \end{bmatrix} = I\label{eq:sls_affine_subspace_constraints_optimization}\,,
    &&\st~ \begin{bmatrix}
        I-Z\mathcal{A} & -Z\mathcal{B}
    \end{bmatrix}
    \begin{bmatrix}
    \bm{\Phi}_x\\\bm{\Phi}_u
    \end{bmatrix} = I\,,\\
    &\qquad\qquad\quad \bm{\Phi}_x, \bm{\Phi}_u \text{ with causal sparsities,}
    &&\qquad\qquad\quad \bm{\Phi}_x, \bm{\Phi}_u \text{ with causal sparsities.}\label{eq:sls_causal_sparsities_constraints_optimization}
\end{alignat}
Here, $\norm{\cdot}_F$ and $\norm{\cdot}_{2 \rightarrow 2}$ denote the Frobenius and the induced 2-norm of a matrix, respectively. Thanks to convexity, the optimal system responses $\{\bm{\Phi}^\star_x, \bm{\Phi}^\star_u\}$ can be computed efficiently and the corresponding optimal control policy in the form of \eqref{eq:control_policy} can then be recovered by $\mathbf{K}^\star = \bm{\Phi}^\star_u{\bm{\Phi}^\star_x}^{-1}$.

\section{Proposed Methodology}
\label{sec:methodology}
In this section, we first adapt useful results from \cite{hassibi1999indefinite}, and we then present a novel tractable method to design regret-optimal control policies that further comply with hard safety constraints over the system states and inputs. Our idea is to exploit the SLP of achievable closed-loop responses to characterize as solutions of convex optimization problems both the clairvoyant optimal policy, which solves the inner minimization in \eqref{eq:dynamic_regret_definition}, and a safe regret-optimal causal controller.
\subsection{The Clairvoyant Optimal Controller}
\label{subsec:clairvoyant_optimal_controller}
In the full-information setting considered in this paper, there exists a unique noncausal control law that outperforms any other controller for every disturbance realization – the clairvoyant policy. To show this, let $\mathcal{F}$ and $\mathcal{G}$ denote the causal response operators comprising the Markov parameters that encode the linear dynamics \eqref{eq:state_dynamics} as $\mathbf{x} = \mathcal{F}\mathbf{u} + \mathcal{G}\mathbf{w}$. Then, the control cost \eqref{eq:lqr_cost} can be expressed as
\begin{align*}
    \operatorname{cost}(\mathbf{w}, \mathbf{u}) &= \left(\mathcal{F} \mathbf{u} + \mathcal{G} \mathbf{w} \right)^\top \mathcal{Q} \left(\mathcal{F} \mathbf{u} + \mathcal{G} \mathbf{w} \right) + \mathbf{u}^\top \mathcal{R} \mathbf{u} = \mathbf{u}^\top \mathcal{P} \mathbf{u} + 2 \mathbf{u}^\top \mathcal{F}^\top \mathcal{Q} \mathcal{G} \mathbf{w} + \mathbf{w}^\top \mathcal{G}^\top \mathcal{Q} \mathcal{G} \mathbf{w}\,,
\end{align*}
where $\mathcal{P} = \mathcal{R} + \mathcal{F}^\top \mathcal{Q}\mathcal{F} \succ 0$. Observing that $\mathcal{Q}\mathcal{F}\mathcal{P}^{-1}\mathcal{F}^\top \mathcal{Q} + \mathcal{Q}(I + \mathcal{F}\mathcal{R}^{-1}\mathcal{F}^\top\mathcal{Q})^{-1} = \mathcal{Q}$ thanks to the Woodbury matrix identity, the incurred control cost can further be written as
\begin{align*}
    \operatorname{cost}(\mathbf{w}, \mathbf{u}) = [\mathcal{P}\mathbf{u} + \mathcal{F}^\top \mathcal{Q} \mathcal{G} \mathbf{w}]^\top 
    \mathcal{P}^{-1}
    [\mathcal{P}\mathbf{u} + \mathcal{F}^\top \mathcal{Q} \mathcal{G} \mathbf{w}] + \mathbf{w}^\top \mathcal{G}^\top \mathcal{Q} (I + \mathcal{F}\mathcal{R}^{-1}\mathcal{F}^\top \mathcal{Q})^{-1} \mathcal{G} \mathbf{w}\,,
\end{align*}
so to highlight the presence of a first non-negative term and of a second addend that does not depend on $\mathbf{u}$. Solving for the cost-minimizing $\mathbf{u}$ by setting the former term equal to zero, we obtain the following input-output description of the clairvoyant controller:
\begin{equation}
    \label{eq:clairvoyant_controller_definition}
    \mathbf{u}^\star = \argmin_{\mathbf{u}^\prime \in \mathbb{R}^{mT}} ~ \operatorname{cost}\left(\mathbf{w}, \mathbf{u}^{\prime}\right) = -(\mathcal{R} + \mathcal{F}^\top\mathcal{Q} \mathcal{F})^{-1}\mathcal{F}^\top\mathcal{Q}\mathcal{G}\mathbf{w}\,.
\end{equation}
Moreover, the control cost suffered by this offline optimal policy, as a function of the sampled disturbance sequence $\mathbf{w}$, is given by
\begin{equation}
    \label{eq:clairvoyant_controller_cost_incurred}
     \operatorname{cost}\left(\mathbf{w}, \mathbf{u}^\star\right) = \mathbf{w}^\top \mathcal{G}^\top \mathcal{Q}(I + \mathcal{F}\mathcal{R}^{-1}\mathcal{F}^\top\mathcal{Q})^{-1} \mathcal{G} \mathbf{w}\,.
\end{equation}
Note that, despite not having specified a priori any parametric structure on such offline policy, the globally optimal dynamic sequence of control actions selected in hindsight by the clairvoyant policy can be expressed as a noncausal linear function of past and future disturbances \citep{hassibi1999indefinite}.

A state-space description of the clairvoyant controller \eqref{eq:clairvoyant_controller_definition} has recently been derived in \cite{goel2020power} and in \cite{foster2020logarithmic} via dynamic programming. Our first result is that the optimal clairvoyant controller \eqref{eq:clairvoyant_controller_definition} can also be computed by solving a convex SLS problem that imposes no causal constraints on the structure of $\{\bm{\Phi}_x, \bm{\Phi}_u\}$. As we establish in Section~\ref{subsec:extensions}, the formulation we propose allows one to explicitly include safety requirements in the definition of the benchmark clairvoyant policy.
\begin{lemma}
    \label{le:clairvoyant_controller_optimization}
    The closed-loop responses $\{\bm{\Phi}_x^{\textit{nc}}, \bm{\Phi}_u^{\textit{nc}}\}$ associated with the optimal noncausal controller $\mathbf{u^\star}$ defined in \eqref{eq:clairvoyant_controller_definition} can be computed as
    \begin{alignat}{3}
        \label{eq:clairvoyant_controller_optimization}
        \{\bm{\Phi}_x^{\textit{nc}}, \bm{\Phi}_u^{\textit{nc}}\} = &~\argmin_{\bm{\Phi}_x, \bm{\Phi}_u} ~ \norm{
        \begin{bmatrix}
            \mathcal{Q}^{\frac{1}{2}} & 0\\
            0 & \mathcal{R}^{\frac{1}{2}}
        \end{bmatrix}
        \begin{bmatrix}
            \bm{\Phi}_x\\\bm{\Phi}_u
        \end{bmatrix}
        \bm{\Sigma}_w^{\frac{1}{2}}
        }_{F}^2
        \in 
        &&~\argmin_{\bm{\Phi}_x, \bm{\Phi}_u} ~ \norm{
        \begin{bmatrix}
            \mathcal{Q}^{\frac{1}{2}} & 0\\
            0 & \mathcal{R}^{\frac{1}{2}}
        \end{bmatrix}
        \begin{bmatrix}
            \bm{\Phi}_x\\\bm{\Phi}_u
        \end{bmatrix}
        }_{2 \rightarrow 2}^2\\
        &\st ~\eqref{eq:sls_affine_subspace_constraints_optimization}\,, &&\st ~\eqref{eq:sls_affine_subspace_constraints_optimization}\nonumber\,.
    \end{alignat}
    Moreover, the control cost \eqref{eq:clairvoyant_controller_cost_incurred} incurred by the noncausal clairvoyant controller on a specific disturbance vector $\mathbf{w}$ is given by
    \begin{equation}
        \label{eq:clairvoyant_controller_cost_optimization}
        \operatorname{cost}\left(\mathbf{w}, \mathbf{u}^\star\right) = \norm{
        \begin{bmatrix}
            \mathcal{Q}^{\frac{1}{2}} & 0\\
            0 & \mathcal{R}^{\frac{1}{2}}
        \end{bmatrix}
        \begin{bmatrix}
            \bm{\Phi}^{\textit{nc}}_x\\\bm{\Phi}^{\textit{nc}}_u
        \end{bmatrix}
        \mathbf{w}
        }_{2}^2\,.
    \end{equation}
\end{lemma}
\begin{proof}
    Observe that the objective functions of the optimization problems in \eqref{eq:clairvoyant_controller_optimization} are equivalent to $\mathbb{E}_{\mathbf{w}}[\operatorname{cost}(\mathbf{w}, \mathbf{u})]$ and $\max_{\|\mathbf{w}\|_{2} \leq 1} ~ \operatorname{cost}(\mathbf{w}, \mathbf{u})$, respectively. Recall that the clairvoyant control law $\mathbf{u^\star}$ defined in \eqref{eq:clairvoyant_controller_definition} is the unique offline policy that minimizes $\operatorname{cost}(\mathbf{w}, \mathbf{u})$ for every $\mathbf{w}$. Hence, for every control sequence $\mathbf{u}$, linearity of the expectation operator implies that
    \begin{equation*}
        \mathbb{E}_{\mathbf{w}}[\operatorname{cost}(\mathbf{w}, \mathbf{u})] - \mathbb{E}_{\mathbf{w}}[\operatorname{cost}(\mathbf{w}, \mathbf{u^\star})] = \mathbb{E}_{\mathbf{w}} [\operatorname{cost}(\mathbf{w}, \mathbf{u}) - \operatorname{cost}(\mathbf{w}, \mathbf{u}^\star)] \geq 0\,.
    \end{equation*}
    Similarly, $\max_{\|\mathbf{w}\|_{2} \leq 1} ~ \operatorname{cost}(\mathbf{w}, \mathbf{u}) \geq \max_{\|\mathbf{w}\|_{2} \leq 1} ~ \operatorname{cost}(\mathbf{w}, \mathbf{u^\star})$ since $\operatorname{cost}(\mathbf{w}, \mathbf{u}) \geq \operatorname{cost}(\mathbf{w}, \mathbf{u^\star})$ point-wise for every $\mathbf{w}$ and every $\mathbf{u}$.
    As the clairvoyant policy $\mathbf{u^\star}$ minimizes both $\mathbb{E}_{\mathbf{w}}[\operatorname{cost}(\mathbf{w}, \mathbf{u})]$ and $\max_{\|\mathbf{w}\|_{2} \leq 1} ~ \operatorname{cost}(\mathbf{w}, \mathbf{u})$, we deduce that the corresponding closed-loop responses $\{\bm{\Phi}_x^{\textit{nc}}, \bm{\Phi}_u^{\textit{nc}}\}$ belong to the set of minimizers of both optimization problems in \eqref{eq:clairvoyant_controller_optimization}. Moreover, we note that $\norm{\mathcal{Q}^{\frac{1}{2}} \bm{\Phi}_x \bm{\Sigma}_w^{\frac{1}{2}}}_{F}^2 + \norm{\mathcal{R}^{\frac{1}{2}} \bm{\Phi}_u \bm{\Sigma}_w^{\frac{1}{2}}}_{F}^2$, with $\bm{\Phi}_x = (I-Z\mathcal{A})^{-1}(I + Z\mathcal{B}\bm{\Phi}_u)$ as per \eqref{eq:sls_affine_subspace_constraints},
    is strictly convex in $\bm{\Phi}_u$ thanks to $\mathcal{Q} \succeq 0\,, \mathcal{R} \succ 0$ and $\bm{\Sigma}_w \succ 0$. Hence, the pair $\{\bm{\Phi}_x^{\textit{nc}}, \bm{\Phi}_u^{\textit{nc}}\}$ constitutes the unique global minimizer of the optimization problem on the left-hand side of \eqref{eq:clairvoyant_controller_optimization}. Lastly, \eqref{eq:clairvoyant_controller_cost_optimization} follows by substituting $\mathbf{x}^{\textit{nc}} = \bm{\Phi}_x^{\textit{nc}}\mathbf{w}$ and $\mathbf{u}^{\textit{nc}} = \bm{\Phi}_u^{\textit{nc}}\mathbf{w}$ in \eqref{eq:lqr_cost}.
\end{proof}
Note that the optimization problems in \eqref{eq:clairvoyant_controller_optimization} mirror the classical $\mathcal{H}_2$ and $\mathcal{H}_\infty$ control formulations presented in Section \ref{subsec:system_level_synthesis}. Indeed, the only difference is the absence of the sparsity constraints \eqref{eq:sls_causal_sparsities_constraints_optimization} on the noncausal closed-loop responses $\{\bm{\Phi}_x^{\textit{nc}}, \bm{\Phi}_u^{\textit{nc}}\}$ corresponding with the offline optimal policy.

Linearity of the clairvoyant policy \eqref{eq:clairvoyant_controller_definition} with respect to the disturbance realizations is central to the regret-optimal synthesis approach proposed in \cite{goel2020regret} in the absence of safety requirements. However, finding the proper change of variables that allows to reduce a regret-suboptimal control problem to a suboptimal $\mathcal{H}_\infty$ problem might be nontrivial, as it requires to analytically characterize a causal factorization of $\gamma^2I + \mathcal{G}^\top \mathcal{Q}\left(I + \mathcal{F} \mathcal{R}^{-1} \mathcal{F}^\top \mathcal{Q} \right)^{-1}\mathcal{G}$, where $\gamma$ denotes the performance level to be tuned iteratively. Indeed, the method of \cite{goel2020regret} involves repeated applications of the whitening property of Kalman filters \citep{kailath2000linear}.

\subsection{The Safe Regret-Optimal Controller}
\label{subsec:safe_regret_optimal_controller}
We now show that linearity of the optimal offline policy \eqref{eq:clairvoyant_controller_optimization} is also key to computing a safe regret-optimal controller with efficient numerical programming techniques. Formulating the nested minimization of \eqref{eq:dynamic_regret_definition} as a convex optimization problem constitutes our main result.
\begin{theorem}
    \label{th:safe_regret_optimal_control}
    Let Assumption~1 hold.
    Consider the evolution of the linear time-varying system \eqref{eq:state_dynamics} over a horizon of length $T \in \mathbb{N}$, and the constrained regret-optimal control problem:
    \begin{subequations}
    \label{prob:safe_regret_optimal_original}
    \begin{alignat}{3}
    \{\bm{\Phi}_x^{\textit{sr}}, \bm{\Phi}_u^{\textit{sr}}\} \in
    &~\argmin_{\bm{\Phi}_x,\bm{\Phi}_u} ~ && \hspace{-0.25cm} \operatorname{regret}(T, \mathbf{u}) \label{eq:safe_regret_optimal_problem_original_th3}\\
        &\st ~&&\eqref{eq:sls_affine_subspace_constraints_optimization} - \eqref{eq:sls_causal_sparsities_constraints_optimization}\,,\nonumber\\
        & &&\mathcal{H} 
        \begin{bmatrix}
            \bm{\Phi}_x^\top & \bm{\Phi}_u^\top
        \end{bmatrix}^\top
        \mathbf{w} \leq \mathbf{h}\,, ~ \forall \mathbf{w} : \mathcal{H}_w \mathbf{w} \leq \mathbf{h}_w\label{eq:safety_constraints_compact_Phi}\,.
    \end{alignat}
    \end{subequations}
    Then, \eqref{prob:safe_regret_optimal_original} is equivalently formulated as the following convex optimization problem:
    \begin{subequations}
    \label{prob:safe_regret_optimal_sdp}
    \begin{alignat}{3}
        \{\bm{\Phi}_x^{\textit{sr}}, \bm{\Phi}_u^{\textit{sr}}\} \in &~\argmin_{\bm{\Phi}_x,\bm{\Phi}_u, \mathbf{Z}, \lambda} ~&&\lambda\label{eq:safe_regret_optimal_problem_sdp_th3}\\
        &\st~ &&\eqref{eq:sls_affine_subspace_constraints_optimization} - \eqref{eq:sls_causal_sparsities_constraints_optimization}\,, ~ \lambda > 0\,, \nonumber\\
        & &&\mathbf{Z}^\top \mathbf{h}_w \leq \mathbf{h}\,, ~
        \mathcal{H}
        \begin{bmatrix}
            \bm{\Phi}_x\\\bm{\Phi}_u
        \end{bmatrix} = \mathbf{Z}^\top \mathcal{H}_w\,, ~ \mathbf{Z}_{ij} \geq 0\,,\label{eq:sdp_dual_safety_constraints}\\
        & &&
        \begin{bmatrix}
            I & & \begin{bmatrix}
            \mathcal{Q}^{\frac{1}{2}} & 0\\
            0 & \mathcal{R}^{\frac{1}{2}}
        \end{bmatrix}\begin{bmatrix}
            \bm{\Phi}_x\\\bm{\Phi}_u
        \end{bmatrix}\\
        \begin{bmatrix}
            \bm{\Phi}_x\\\bm{\Phi}_u
        \end{bmatrix}^\top
        \begin{bmatrix}
            \mathcal{Q}^{\frac{1}{2}} & 0\\
            0 & \mathcal{R}^{\frac{1}{2}}
        \end{bmatrix} & & \lambda I + \begin{bmatrix}
            \bm{\Phi}_x^{\textit{nc}}\\\bm{\Phi}_u^{\textit{nc}}
        \end{bmatrix}^{\top}
        \begin{bmatrix}
            \mathcal{Q} & 0\\
            0 & \mathcal{R}
        \end{bmatrix}
        \begin{bmatrix}
            \bm{\Phi}_x^{\textit{nc}}\\\bm{\Phi}_u^{\textit{nc}}
        \end{bmatrix}
        \end{bmatrix} \succeq 0\label{eq:sdp_schur_constraints}\,.
    \end{alignat}
    \end{subequations}
\end{theorem}
\begin{proof}
    In light of \eqref{eq:clairvoyant_controller_cost_optimization}, \eqref{prob:safe_regret_optimal_original}
    can be equivalently expressed as
    \begin{align}
        &~\min_{\bm{\Phi}_x, \bm{\Phi}_u}  ~\max_{\|\mathbf{w}\|_{2} \leq 1}
        ~ \mathbf{w}^\top 
        \underbrace{
        \left(
        \begin{bmatrix}
            \bm{\Phi}_x\\\bm{\Phi}_u
        \end{bmatrix}^{\top}
        \begin{bmatrix}
            \mathcal{Q} & 0\\
            0 & \mathcal{R}
        \end{bmatrix}
        \begin{bmatrix}
            \bm{\Phi}_x\\\bm{\Phi}_u
        \end{bmatrix} 
        -
        \begin{bmatrix}
            \bm{\Phi}_x^{\textit{nc}}\\\bm{\Phi}_u^{\textit{nc}}
        \end{bmatrix}^{\top}
        \begin{bmatrix}
            \mathcal{Q} & 0\\
            0 & \mathcal{R}
        \end{bmatrix}
        \begin{bmatrix}
            \bm{\Phi}_x^{\textit{nc}}\\\bm{\Phi}_u^{\textit{nc}}
        \end{bmatrix} 
        \right)
        }_{\Delta(\bm{\Phi}_x, \bm{\Phi}_u)}
        \mathbf{w}
        \label{eq:dynamic_regret_difference_quadratic_form_delta_definition}\\
        &\st ~\eqref{eq:sls_affine_subspace_constraints_optimization} - \eqref{eq:sls_causal_sparsities_constraints_optimization}\,, ~ \eqref{eq:safety_constraints_compact_Phi}\nonumber\,.
    \end{align}
    
    We first proceed to derive a tractable formulation for the objective function in \eqref{eq:dynamic_regret_difference_quadratic_form_delta_definition}. Since the offline optimal controller attains minimum control cost on every perturbation sequence, we have that $\Delta(\bm{\Phi}_x, \bm{\Phi}_u) \succeq 0$ for all $\{\bm{\Phi}_x, \bm{\Phi}_u\}$. Consequently, there exists a matrix $\delta(\bm{\Phi}_x, \bm{\Phi}_u) \succeq 0$ such that $\Delta = \delta^\top \delta$,\footnote{Such a factorization can be computed from the eigendecomposition $\Delta = V \Lambda V^{-1}$. Since $\Delta$ is positive semidefinite, its eigenvalues are non-negative and its eigenvectors are orthogonal, hence $\Delta = V \Lambda V^\top = V \Lambda^{\frac{1}{2}} (V \Lambda^{\frac{1}{2}})^{\top} = \delta^\top \delta$.} where we have omitted the dependence of $\Delta$ and $\delta$ on $\{\bm{\Phi}_x, \bm{\Phi}_u\}$ to ease readability. Then, leveraging well-known properties of induced matrix norms, we obtain
    \begin{equation*}
        \max_{\|\mathbf{w}\|_{2} \leq 1}
        ~ \mathbf{w}^\top \Delta \mathbf{w} = \max_{\|\mathbf{w}\|_{2} \leq 1} ~ \norm{\delta \mathbf{w}}_2^2 = \norm{\delta}_{2 \rightarrow 2}^2 = \sigma^2_{\operatorname{max}}\left(\delta\right) = \lambda_{\operatorname{max}}\left(\Delta\right)\,,
    \end{equation*}
    where $\sigma_{\operatorname{max}}(\cdot)$ and $\lambda_{\operatorname{max}}(\cdot)$ denote the largest singular value and the largest eigenvalue of a matrix, respectively. Hence, minimizing dynamic regret corresponds to minimizing the maximum eigenvalue of $\Delta$, which depends quadratically on $\{\bm{\Phi}_x, \bm{\Phi}_u\}$. Building upon classical results on semidefinite programming for eigenvalue minimization (see, e.g., Section~2.2 in \cite{boyd1994linear}), $\lambda_{\operatorname{max}}(\Delta)$ is equivalently computed as $\min_{\lambda} \lambda ~ \st ~ \lambda I - \Delta \succeq 0\,, ~ \lambda > 0$. Recalling the definition of $\Delta$ in \eqref{eq:dynamic_regret_difference_quadratic_form_delta_definition} and exploiting the Schur complement, we obtain the desired expression \eqref{eq:sdp_schur_constraints}.
    
    For the safety constraints, we apply dualization to eliminate the universal quantifier from \eqref{eq:safety_constraints_compact_Phi}. In particular, we recognize that each row of $\max_{\mathbf{w} \in \mathbb{W}} ~ \mathcal{H} \begin{bmatrix}
        \bm{\Phi}_x^\top & \bm{\Phi}_u^\top
    \end{bmatrix}^\top  \mathbf{w}$ is a linear optimization problem \citep{goulart2007affine}. Exploiting duality, we reformulate the constraints as
    \begin{alignat*}{3}
        \max_{\mathbf{w} \in \mathbb{W}} ~ \left(
        \mathcal{H}
        \begin{bmatrix}
            \bm{\Phi}_x\\
            \bm{\Phi}_u
        \end{bmatrix}\right)_i \mathbf{w} = &~\min_{\mathbf{z}_i} ~ \mathbf{h}_w^\top \mathbf{z}_i\,,\\ 
        &\st ~\mathbf{z}_i \geq 0\,, ~ \mathcal{H}_w^\top \mathbf{z}_i = \left(
        \mathcal{H}
        \begin{bmatrix}
            \bm{\Phi}_x\\
            \bm{\Phi}_u
        \end{bmatrix}
        \right)_i^\top\,, 
    \end{alignat*}
    where the vector $\mathbf{z}_i$ represents the dual vector associated
    with the $i$-th row of the maximization. Finally, by combining the dual variables that arise from each row into the matrix $\mathbf{Z}$, we obtain the linear constraints \eqref{eq:sdp_dual_safety_constraints}.
\end{proof}
We note that the method in \cite{goel2020regret} allows computing the unconstrained regret-optimal controller by solving a series of Riccati recursions. Instead, our approach to the more general constrained case relies on solving \eqref{prob:safe_regret_optimal_sdp} through semidefinite programming, which poses challenges to scalability, in general. We refer the interested reader to \cite{ahmadi2019dsos, zheng2017fast} and the references therein for state-of-the-art techniques that exploit diagonal dominance and chordal sparsity to improve scalability. Last, notice that a safe regret-optimal policy in the form of \eqref{eq:control_policy} is recovered by $\mathbf{K}^{\textit{sr}} = \bm{\Phi}_u^{\textit{sr}} {\bm{\Phi}_x^{\textit{sr}}}^{-1}$.

\subsection{Competing Against Safety-Aware Control Benchmarks}
\label{subsec:extensions}
The safe regret-optimal control law characterized in Theorem~\ref{th:safe_regret_optimal_control} competes against the unconstrained clairvoyant policy \eqref{eq:clairvoyant_controller_optimization}. Alternatively, one could aim to minimize the loss relative to an optimal safety-aware noncausal policy that selects the control actions with complete knowledge of past and future disturbances, but has to obey the same safety constraints as the online controller. We note that a similar idea has recently been studied in \cite{li2021safe}, which aims to bound the regret with respect to the best safe state-feedback causal controller in hindsight. Instead, our proposal is to synthesize safe control laws that minimize regret against safe clairvoyant policies through convex programming techniques. Specifically, we show that the optimization perspective of Lemma~\ref{le:clairvoyant_controller_optimization} and Theorem~\ref{th:safe_regret_optimal_control} can also be used for defining safety-aware control benchmarks and, subsequently, safe regret-optimal policies that compete against them. 
\begin{corollary}
    \label{cor:safe_regret_safe_benchmark}
    Let Assumption~1 hold. The closed-loop responses $\{\bm{\Phi}_x^{\textit{snc}}, \bm{\Phi}_u^{\textit{snc}}\}$ associated with a safe clairvoyant policy that is optimal either in the $\mathcal{H}_2$ or in the $\mathcal{H}_\infty$ sense are computed by solving
    \begin{alignat}{3}
        &~\min_{\bm{\Phi}_x,\bm{\Phi}_u, \mathbf{Z}} ~ \norm{
        \begin{bmatrix}
            \mathcal{Q}^{\frac{1}{2}} & 0\\
            0 & \mathcal{R}^{\frac{1}{2}}
        \end{bmatrix}
        \begin{bmatrix}
            \bm{\Phi}_x\\\bm{\Phi}_u
        \end{bmatrix}
        \bm{\Sigma}_w^{\frac{1}{2}}
        }_{F}^2
        \hspace{0.75cm}\text{and}\hspace{0.75cm}
        &&~\min_{\bm{\Phi}_x,\bm{\Phi}_u, \mathbf{Z}} ~ \norm{
        \begin{bmatrix}
            \mathcal{Q}^{\frac{1}{2}} & 0\\
            0 & \mathcal{R}^{\frac{1}{2}}
        \end{bmatrix}
        \begin{bmatrix}
            \bm{\Phi}_x\\\bm{\Phi}_u
        \end{bmatrix}
        }_{2 \rightarrow 2}^2
        \label{eq:safe_benchmark_optimal_H2_Hinfty}\\
        &\st~ \eqref{eq:sls_affine_subspace_constraints_optimization}\,, ~ \eqref{eq:sdp_dual_safety_constraints}\,,
        &&\st~ \eqref{eq:sls_affine_subspace_constraints_optimization}\,, ~ \eqref{eq:sdp_dual_safety_constraints}\nonumber\,,
    \end{alignat}
    respectively. Moreover, a safe regret-optimal policy that competes against a safe clairvoyant controller is computed by solving \eqref{prob:safe_regret_optimal_sdp}, provided that 
    $\{\bm{\Phi}_x^{\textit{nc}}, \bm{\Phi}_u^{\textit{nc}}\}$ is replaced by $\{\bm{\Phi}_x^{\textit{snc}}, \bm{\Phi}_u^{\textit{snc}}\}$ in \eqref{eq:sdp_schur_constraints}.
\end{corollary}
Note that Assumption~1 guarantees that both optimization problems in \eqref{eq:safe_benchmark_optimal_H2_Hinfty} are feasible, since their solution spaces include that of \eqref{prob:safe_regret_optimal_sdp}.
We recall that in the unconstrained setting both cost formulations in \eqref{eq:safe_benchmark_optimal_H2_Hinfty} were minimized by the same noncausal input sequence  \eqref{eq:clairvoyant_controller_optimization} as per Lemma~\ref{le:clairvoyant_controller_optimization}. Conversely, in the constrained case there may not exist a safe clairvoyant control law that outperforms all other safe policies on every disturbance instance. Hence, the minimizers of the two optimization problems in \eqref{eq:safe_benchmark_optimal_H2_Hinfty} may be different, and their optimality should only be understood in the $\mathcal{H}_2$ or $\mathcal{H}_\infty$ sense. This observation is consistent with the partial-information setup studied in \citep{goel2021regret}.

The results of Theorem~\ref{th:safe_regret_optimal_control} and Corollary~\ref{cor:safe_regret_safe_benchmark} naturally extend to more complex settings. For instance, the proposed optimization standpoint allows one to constrain the system closed-loop responses to lie in any arbitrary set $\mathcal{S}$, provided that it admits a convex representation. In particular, constraints of the form $\{\bm{\Phi}_x, \bm{\Phi}_u\} \in \mathcal{S}$ could model additional performance requirements, structural constraints that arise from the distributed nature of the system under control, or sparsity surrogate requirements.

\section{Numerical Results}
We test the performance of the proposed safe regret-optimal control law $\mathcal{S}\mathcal{R}_\textit{nc}$ against classical constrained $\mathcal{H}_2$ and $\mathcal{H}_\infty$ controllers, which can be computed by solving \eqref{eq:safe_benchmark_optimal_H2_Hinfty} subject to additional sparsity constraints \eqref{eq:sls_causal_sparsities_constraints_optimization}. For our experiments, we consider the evolution of system \eqref{eq:state_dynamics}, starting from the unknown initial condition $x_0 = 0$, with
\begin{equation*}
    A_t = \rho \begin{bmatrix}
        0.7 & 0.2 & 0\\
        0.3 & 0.7 & -0.1\\
        0 & -0.2 & 0.8
    \end{bmatrix}\,, ~
    B_t = \begin{bmatrix}
        1 & 0.2\\
        2 & 0.3\\
        1.5 & 0.5
    \end{bmatrix}\,, ~ \forall t \in \{0 \dots T-1\}\,,
\end{equation*}
where $\rho$ is the spectral radius of the system, and $T = 30$. Letting $\rho = 0.7$, we first compute $\{\bm{\Phi}_x^{\textit{nc}}, \bm{\Phi}_u^{\textit{nc}}\}$ as the solution to \eqref{eq:clairvoyant_controller_optimization}, choosing $\mathcal{Q} = I_{30} \otimes I_3$ and $\mathcal{R} = I_{30} \otimes I_2$ in \eqref{eq:lqr_cost}, where $\otimes$ denotes the Kronecker product. Subsequently, we solve the SDP in \eqref{prob:safe_regret_optimal_sdp} to synthesize a control policy that minimizes dynamic regret while complying, for all possible initial conditions $-1 \leq x_0 \leq 1$ and perturbations $-1 \leq w_t \leq 1$, with safety constraints $-3 \leq x_t \leq 3\,, ~ -2 \leq u_t \leq 2\,, ~ \forall t \in \{0 \dots T-1\}$. We simulate the evolution of the closed-loop system perturbed by 1000 different realizations of a disturbance sequence uniformly distributed between $-1$ and $1$, i.e., $w_t \overset{\text{iid}}{\sim} \mathcal{U}_{[-1,1]^3}$, and we verify that the input and state trajectories are safe in all rounds, as expected. 

Then, for the two cases, $\rho = 0.7$ (open-loop stable system) and $\rho = 1.05$ (open-loop unstable system), we compare the average control cost suffered by these safe control laws when the disturbances are drawn according to a variety of stochastic and deterministic profiles. For the case $\rho = 1.05$, we relax the safety constraints and require that $-10 \leq x_t \leq 10\,, ~ -10 \leq u_t \leq 10\,, ~ \forall t \in \{0 \dots T-1\}$ to achieve feasibility. We collect our results in Table~\ref{table:results}.\footnote{The code that reproduces our numerical examples is available at \href{https://github.com/DecodEPFL/SafeMinRegret}{https://github.com/DecodEPFL/SafeMinRegret}. Please refer to the simulation code for a precise definition of the disturbance profiles that appear in Table~\ref{table:results}.} As expected, when the true perturbation sequence follows a Gaussian distribution, i.e., $w_t \overset{\text{iid}}{\sim} \mathcal{N}(0,I_3)$, $\mathcal{SH}_2$ achieves the best performance, closely followed by our safe regret-optimal control law $\mathcal{SR}_{\textit{nc}}$ when $\rho = 0.7$, as observed in \cite{goel2020regret} for the unconstrained case. Similarly, when the disturbance realizations are chosen adversarially, $\mathcal{SR}_{\textit{nc}}$ nearly tracks the performance of $\mathcal{SH}_\infty$, which attains the lowest control cost. Instead, our $\mathcal{SR}_{\textit{nc}}$ consistently outperforms both $\mathcal{SH}_2$ and $\mathcal{SH}_\infty$ in almost all other scenarios. We conjecture that this improvement is linked with the optimality of the control benchmark in \eqref{eq:dynamic_regret_difference_quadratic_form_delta_definition}, which is tailored to the system at hand. Furthermore, we observe that the performance increase may be very significant. This is made clear, for instance, in our experiments with $\rho = 0.7$, where $\mathcal{SH}_2$ and $\mathcal{SH}_\infty$ always incur a loss that is at least $40\%$ higher.
\begin{table}[ht]
\caption{Average control cost increase relative to the control policy denoted with {\color{darkgreen} \textbf{1}}.}
\label{table:results}
\vspace{3pt}
\parbox{.495\linewidth}{
    \centering
    Open-loop stable system: $\rho = 0.7$
    \rowcolors{1}{}{lightgray}
    \begin{tabular}{c|ccc}
        \hline
        $\mathbf{w}$ & $\mathcal{SH}_2$ & $\mathcal{SH}_\infty$ & $\mathcal{SR}_{\textit{nc}}$\\
        \hline
        $\mathcal{N}(0,1)$ & \color{darkgreen} \textbf{1} & +21.14\% & + 10.89\%\\
        $\mathcal{U}_{[0.5, 1]}$ & +63.42\% & $>$+100\% & \color{darkgreen} \textbf{1}\\
        $\mathcal{U}_{[0, 1]}$ & +40.69\% & $>$+100\% & \color{darkgreen} \textbf{1}\\
        $1$ & +67.74\% & $>$+100\% & \color{darkgreen}\textbf{1}\\
        $\operatorname{sin}$ & +58.12\% & $>$+100\% & \color{darkgreen}\textbf{1}\\
        $\operatorname{sawtooth}$ & +46.27\% & $>$+100\% & \color{darkgreen} \textbf{1}\\
        $\operatorname{step}$ & +66.49\% & $>$+100\% & \color{darkgreen} \textbf{1}\\
        $\operatorname{stairs}$ & +45.27\% & $>$+100\% & \color{darkgreen} \textbf{1}\\
        $\operatorname{worst}$ & +18.45\% & \color{darkgreen} \textbf{1} & +7.74\%\\
        \hline
    \end{tabular}
}
\hfill
\parbox{.495\linewidth}{
    \centering
    Open-loop unstable system: $\rho = 1.05$
    \rowcolors{1}{}{lightgray}
    \begin{tabular}{c|ccc}
        \hline
        $\mathbf{w}$ & $\mathcal{SH}_2$ & $\mathcal{SH}_\infty$ & $\mathcal{SR}_{\textit{nc}}$\\
        \hline
        $\mathcal{N}(0,1)$ & \color{darkgreen} \textbf{1} & $>$+100\% & + 51.99\%\\
        $\mathcal{U}_{[0.5, 1]}$ & +36.60\% & +9.91\% & \color{darkgreen} \textbf{1}\\
        $\mathcal{U}_{[0, 1]}$ & +5.60\% & +20.14\% & \color{darkgreen} \textbf{1}\\
        $1$ & +44.57\% & +7.25\% & \color{darkgreen}\textbf{1}\\
        $\operatorname{sin}$ & +39.10\% & +12.89\% & \color{darkgreen}\textbf{1}\\
        $\operatorname{sawtooth}$ & +26.53\% & +13.52\% &\color{darkgreen} \textbf{1}\\
        $\operatorname{step}$ & +15.66\% & \color{darkgreen} \textbf{1} & \color{darkgreen} \textbf{1} \\
        $\operatorname{stairs}$ & +15.37\% & +0.46\% &\color{darkgreen} \textbf{1}\\
        $\operatorname{worst}$ & $>$+100\% &\color{darkgreen} \textbf{1} & +26.51\%\\
        \hline
    \end{tabular}
}
\end{table}

\section{Conclusion}
We have presented a novel method for convex synthesis of regret-optimal control policies that comply with hard safety requirements. To do so, we have first characterized the clairvoyant policy by extending the SLS framework, and we have then nested the corresponding solution into a constrained regret-optimal SLS program. Numerical results show that control laws that safely minimize regret can adapt to heterogeneous disturbance sequences, effectively interpolating between the performance of, or even prevailing over, constrained $\mathcal{H}_2$ and $\mathcal{H}_\infty$ controllers. Future work encompasses extensions to the infinite-horizon case, as well as distributed control and model-free scenarios.

\acks{Research supported by the Swiss National Science Foundation under the NCCR Automation (grant agreement 51NF40\textunderscore 180545).}

\bibliography{references}

\end{document}